\documentclass[letterpaper, 10 pt,conference]{IEEEtran}

\IEEEoverridecommandlockouts   
\usepackage{amsmath}
\usepackage{amsfonts}
\usepackage{amssymb}

\usepackage{amsthm}
\usepackage{setspace}
\usepackage{subfigure}
\usepackage{color}
\usepackage{float}
\usepackage{graphicx}
\usepackage{lettrine}
\usepackage{pgfplots}
\usepackage{comment}
\usepackage{cite} 
\usepackage{mathtools}
\usepackage{longtable,tabularx}
\usepackage{multirow}
\usepackage{algorithm}
\usepackage{algpseudocode}
\let\labelindent\relax
\usepackage{enumitem}

\newlist{steps}{enumerate}{1}
\setlist[steps, 1]{label = Step \arabic*:}
\newtheorem{thm}{Theorem}
\newtheorem{lemma}{Lemma}

\newcommand{\bmtx}{\begin{bmatrix}}
\newcommand{\emtx}{\end{bmatrix}}
\newcommand{\bsmtx}{\left[ \begin{smallmatrix}} 
\newcommand{\esmtx}{\end{smallmatrix} \right]}
\newcommand{\bmatarray}[1]{\left[\begin{array}{#1}}
\newcommand{\ematarray}{\end{array}\right]}

\newcommand{\field}[1]{\mathbb{#1}}
\newcommand{\R}{\field{R}}

\newcommand{\akr}[1]{{\color{black}{#1}}}
\newcommand{\ak}[1]{{\color{black}{#1}}}
\newcommand{\tm}[1]{{\color{black}{#1}}}

\begin{document}

\title{Estimating Regions of Attraction for\\ Transitional Flows using Quadratic Constraints}

\author{\IEEEauthorblockN{Aniketh Kalur,$^1$  Talha Mushtaq,$^1$ Peter Seiler,$^2$ and Maziar S. Hemati$^1$\\ }
\IEEEauthorblockA{$^1$Aerospace Engineering and Mechanics,
University of Minnesota,
Minneapolis, MN 55455, USA\\
$^2$Electrical Engineering and Computer Science,
 University of Michigan, Ann Arbor, MI 48109, USA 
}}


\maketitle  
\thispagestyle{empty} 
\begin{abstract}
  This letter describes a method for estimating regions of attraction and bounds on permissible perturbation amplitudes in nonlinear fluids systems.
  The proposed approach exploits quadratic constraints between the inputs and outputs of the
  nonlinearity on elliptical sets.
  This approach reduces conservatism and improves estimates for regions of attraction and bounds on permissible perturbation amplitudes over related methods that employ quadratic constraints on spherical sets.
  We present and investigate two algorithms for performing the analysis: an iterative method that refines the analysis by solving a sequence of semi-definite programs, and another based on solving a generalized eigenvalue problem with lower computational complexity, but at the cost of some precision in the final solution.
  The proposed algorithms are demonstrated on low-order mechanistic models of transitional flows.
  We further compare accuracy and computational complexity with analysis based on sum-of-squares optimization and direct-adjoint looping methods.
\end{abstract}

\section{Introduction}
\label{sec:intro}
\IEEEPARstart{E}{nvironmental} disturbances can cause fluid flows to transition from a low-skin-friction laminar state to
a high-skin-friction turbulent state
when the Reynolds number~($Re$) is sufficiently large.
Yet, precisely predicting the onset of transition is notoriously difficult, even in the simplest of geometries~\cite{Schmid2001,barkleyJFM2016,Kerswell_2018}.
\akr{An ability to reliably estimate if and when transition will arise 
 is directly related to the problem of identifying the region of attraction~(ROA) of the system.} 
To this end, in this work we investigate systems-theoretic analysis methods for estimating the ROA of a laminar equilibrium flow and for determining associated bounds on permissible perturbation amplitudes for remaining in this ROA.

%
%
Recent efforts for nonlinear stability analysis of the incompressible Navier-Stokes equations~(NSE) have exploited a Lur'e decomposition~\cite{khalil01} of the system dynamics into a feedback interconnection between the non-normal linear dynamics and quadratic energy-conserving nonlinearity.
%
Such approaches include dissipation inequalities~\cite{Ahmadi2018}, passivity analysis~\cite{zhao2013}, and sum-of-squares~(SOS) optimization~\cite{Goulart_2012}, all of which generalize the classical energy-based methods of hydrodynamic stability theory~\cite{Schmid2001,josephBook}.  
\akr{Methods for the analysis of systems with quadratic nonlinearities have also been proposed in prior works~\cite{AmatoAut2007,AmatoConf2007}; however, these methods scale combinatorially with the state dimension, prohibiting their use on high-dimensional fluids systems.}

Most recently, a series of studies have proposed \akr{exploiting} quadratic constraints~(QCs) between the inputs and outputs of the nonlinearity to conduct global and local stability analysis with reduced-complexity~\cite{kalurAIAA2020,kalur21,liu20}. 
%
%
The trade off for this computational expediency is a larger degree of conservatism in estimating the ROA and associated bounds on permissible perturbation amplitudes relative to more computationally demanding methods, such as SOS~\cite{Goulart_2012} and direct-adjoint looping~(DAL)~\cite{Kerswell_2018}. 
The QC formulation in~\cite{liu20} reduces conservatism compared to the approach in~\cite{kalur21}, \akr{but some conservatism remains because of a restriction to spherical sets.} 
%
%

In this work, we generalized the QCs presented in~\cite{kalurAIAA2020,kalur21} and~\cite{liu20} to arbitrary ellipsoidal sets. 
As we will show, these new QCs reduce conservatism and improve estimates of both the ROA and the largest permissible perturbation.
We propose algorithms for performing this analysis: one is an iterative algorithm that solves a semi-definite program at each iteration to refine the ROA estimate, and the other is based on solving a single generalized eigenvalue problem~(GEVP).
Using the QCs generalized on ellipsoidal sets, we analyze ROA estimates and the largest permissible perturbation for system stability; the inner estimate of the ROA captures this perturbation. As an example, we will demonstrate our approach on two low-dimensional mechanistic transitional flow models: the 4-state Walleffe-Kim-Hamilton~(WKH) model of shear flow~\cite{WKH} and the 9-state model of Couette flow~\cite{Moehlis_2004}. Finally, we measure the computational run-time and show that the proposed QC method obtains improved estimates over previous QC approaches~\cite{kalur21,liu20}, while reducing computational time over SOS and DAL methods.

\section{Problem Formulation}
\label{sec:problem}
Consider a nonlinear system of the following form:
\begin{equation}
  \label{eq:nlsys}
  \dot{x}(t) = A x(t) + N( x(t) )
\end{equation}
where $x(t) \in \R^n$ is the state and the state matrix
$A \in \R^{n\times n}$ is Hurwitz. The nonlinearity
$N: \R^n \rightarrow \R^n$ is assumed to be a quadratic function of
the form:
\begin{align}
  \label{eq:Ndef}
  N(x) = \bmtx x^T Q_1 x \\ \vdots \\ x^T Q_n x \emtx
\end{align}
where $Q_1,\ldots,Q_n \in \R^{n\times n}$ are symmetric (but not necessarily sign definite)  matrices. Moreover, the nonlinearity is assumed to be lossless:
$x^T N(x)=0$ $\forall$ $x\in \R^n$. This lossless property is observed in the nonlinear terms of the incompressible NSE and other reduced-order models that mimic transitional flows~\cite{Moehlis_2004,waleffe95}.


It also follows that $N(0)=0$. Hence $\bar{x}=0$ is
an equilibrium point of the nonlinear system \eqref{eq:nlsys}. This is
an asymptotically stable equilibrium point because $A$ is
Hurwitz~(see Theorem. 4.5 in \cite{khalil01}).  Let
$\phi(t,x(0))$ denote the solution of \eqref{eq:nlsys} at time $t$ from
the initial condition $x(0)$. The region of attraction (ROA) for
$\bar{x}=0$ is defined as:
\begin{align}
  \label{eq:ROA}
  \mathcal{R} := \{ x(0) \in \R^n \, : \, \phi(t,x(0)) \to 0 
          \mbox{ as } t\to \infty \}.
\end{align}
In other words, the ROA is the set of initial conditions for which the
trajectory asymptotically converges back to the equilibrium point. 
The
equilibrium point $\bar{x}=0$ is globally asymptotically stable if
$\mathcal{R} = \R^n$.  In general, the equilibrium point will be locally
but not globally asympotitically stable.  The objective is to obtain an 
inner estimate~$\mathcal{\hat{R}}$ of the ROA $\mathcal{R}$, i.e., to compute a set $\mathcal{\hat{R}}
\subset \mathcal{R}$.

\section{Stability Analysis}
\label{sec:stabanalysis}

The stability analysis is based on separating the nonlinearity from
the remaining linear dynamics:
\begin{align}
  \label{eq:lure}
  \dot{x}(t) & = A x(t) + z(t) \\
  z(t) & = N( x(t) ).
\end{align}
This system can be represented as the Lur'e decomposition \cite{khalil01} as shown in Figure~\ref{fig:lure}.

\begin{figure}[h]
\centering
\scalebox{0.9}{
\begin{picture}(120,90)(-20,0)
 \thicklines
 \put(0,0){\framebox(70,40){$\dot{x}=Ax+z$}}
 \put(15,50){\framebox(40,40){$N(x)$}}
 \put(-30,45){$x$}
 \put(0,20){\line(-1,0){20}}  
 \put(-20,20){\line(0,1){50}}  
 \put(-20,70){\vector(1,0){35}}  
 \put(94,45){$z$}
 \put(55,70){\line(1,0){35}}  
 \put(90,70){\line(0,-1){50}}  
 \put(90,20){\vector(-1,0){20}}  
\end{picture}
} 
\caption{Lur'e decomposition of nonlinear system.}
\label{fig:lure}
\end{figure}
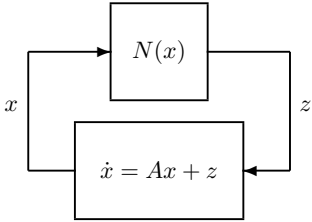

\subsection{Local Quadratic Constraints}
\label{sec:constraints}

The input-output properties of the nonlinearity can be bounded using a set of
QCs on $(x,z)$.  The lossless property yields the
following global QC:
\begin{align}
  \bmtx x \\ z \emtx^T \bmtx 0 & I \\ I & 0 \emtx
  \bmtx x \\ z \emtx = 0
  \,\,\, \forall x \in \R^n, z=N(x).
  \label{eq:Lossless_Constraint}
\end{align}
To study the effects of nonlinearities locally, additional QCs were formulated in \cite{kalur21} and in \cite{liu20}. The local QCs in both of these works were defined on a spherical set. 
In this work, we reduce the conservatism of the aforementioned approaches by generalizing to constraints on an ellipsoidal set. 
The next lemma generalizes the
result in \cite{liu20} to provide local constraints on an ellipsoidal set.

\begin{lemma}
\label{lem:localqc}
Let $E=E^T \succ 0$ be given and define the ellipsoid
$\mathcal{E}_\alpha := \{ x\in \R^n \, : \, x^T E x \le \alpha^2\}$.
The nonlinearity $N$ given in \eqref{eq:Ndef} satisfies
the following local QC for $i=1,\ldots,n$:
\begin{align}
  & \bmtx x \\ z \emtx^T \bmtx \alpha^2 (Q_i E^{-1} Q_i) & 0 
         \\ 0 & -e_i e_i^T \emtx
  \bmtx x \\ z \emtx \ge 0,~\forall x \in \mathcal{E}_\alpha, 
  \label{eq:Local_Constraints}
\end{align}  
where $e_i\in \R^n$ is the $i^{th}$ standard basis vector.
\end{lemma}
\begin{proof}
Note that $z^T e_i e_i^T z = z_i^2$, where $z_i:=x^T Q_i x$ is the $i^{th}$ entry of $z=N(x)$.  Define $w:=E^{\frac{1}{2}} x$ and $\hat{Q}_i:= E^{-\frac{1}{2}} Q_i E^{-\frac{1}{2}}$ so that $z_i = w^T \hat{Q}_i w$.
 The Cauchy-Schwartz inequality yields
  the following bound:
  \begin{align}
    \label{eq:zisq}
    z_i^2 \le \|w\|_2^2 \cdot \| \hat{Q}_i w\|_2^2.
  \end{align}
  Note that $\| \hat{Q}_i w\|_2^2 = x^T Q_i  E^{-1} Q_i x$. Moreover,
  if $x \in \mathcal{E}_\alpha$ then $\|w\|_2^2 = x^T E x \le \alpha^2$.
  Combining these facts with Eq.~\eqref{eq:zisq} yields
$z_i^2 \le x^T [\alpha^2 (Q_i E^{-1} Q_i)] x$
  for any
  $x \in \mathcal{E}_\alpha$.
\end{proof}
This result corresponds to Lemma 1 in \cite{liu20} for the special
case $E=I$. This special case corresponds to a local constraint on a
sphere of radius $\alpha$. The generalization to local constraints on
arbitrary ellipsoids will be used to improve our estimates of the ROA.

\subsection{ROA Estimation}
\label{sec:ROA}

We can combine Lyapunov theory with the local QCs from the previous section in order to compute an inner estimate $\hat{\mathcal{R}}$ for the ROA.  Roughly, we will define a Lyapunov candidate  $V(x)=x^T P x$ and use the QCs to show that $\dot{V}$ is negative definite along the trajectories of Eq.~\eqref{eq:nlsys} in a neighborhood of the equilibrium point $\bar{x}=0$.  The inner estimate of the ROA will be given by a sphere of radius $R$, denoted $\hat{\mathcal{R}}_R:=\{ x\in \R^n \, : \,  x^T x \le R^2\}$. The next theorem gives a matrix inequality condition to estimate the ROA using local QCs.  This is based on a standard Lyapunov result (Theorem 4.1 in [5]). To simplify notation, define the following matrices that appear in the QCs:

\begin{align}
\label{eq:constraint_matrix}
M_0 := \bmtx 0 & I \\ I & 0 \emtx,
\akr{M_i(\alpha,E)} := \bmtx \alpha^2  (Q_i E^{-1} Q_i) & 0 
\\ 0 & -e_i e_i^T \emtx. 
\end{align}

\begin{thm}
\label{thm:roaest}
Let $E=E^T \succ 0$, $\alpha>0$, $\epsilon>0$ be given. If
$\exists P=P^T \in \R^{n\times n}$, $R>0$, and $\xi_0,\ldots, \xi_n\in \R$
such that:
\begin{align}
\label{eq:roaLMI}
& \bmtx A^TP +P A & P \\ P & 0 \emtx + \xi_0 M_0 + \sum_{i=1}^n \xi_i \akr{M_i(\alpha,E)}
 \preceq \bmtx -\epsilon I & 0 \\ 0 & 0 \emtx  \\
\label{eq:roaPineq}
& \frac{1}{\alpha^2} E \preceq P \preceq \frac{1}{R^2} I \\
\label{eq:roaxi}
& \xi_i \ge 0 \mbox{ for } i=1,\ldots,n
\end{align}
then $\hat{\mathcal{R}}_R \subset \mathcal{R}$.
\end{thm}
\begin{proof}
  Define the Lyapunov function $V(x):=x^T P x$. Note that
  $\frac{1}{\alpha^2} E \preceq P$ implies $P \succ 0$. Multiply
  \eqref{eq:roaLMI} on the left/right by $\bmtx x(t)^T & z(t)^T \emtx$ and
  its transpose to obtain:
  \begin{align*}
    & \frac{d}{dt} V( x(t) ) 
    + \xi_0 \bmtx x(t) \\ z(t) \emtx^T M_0 \bmtx x(t) \\ z(t) \emtx \\
    & + \sum_{i=1}^n \xi_i \bmtx x(t) \\ z(t) \emtx^T 
           M_i \bmtx x(t) \\ z(t) \emtx
    \le -\epsilon \| x(t) \|_2^2.
  \end{align*}
  The second term with $\xi_0$ and $M_0$ is equal to zero due to the
  global lossless property of $N(x)$. \akr{Here, the scalar term $\xi_0$ can be either positive or negative. While }the quadratic terms with $\xi_i$
  and $M_i$ ($i=1$ to $n$) are each non-negative for any $x(t)\in \mathcal{E}_\alpha$ by
  Lemma~\ref{lem:localqc} and $\xi_i\ge 0$. Thus
  $x(t)\in \mathcal{E}_\alpha$ implies
  $\frac{d}{dt} V( x(t) ) \le -\epsilon \| x(t) \|_2^2$.

  The constraint $\frac{1}{\alpha^2} E \preceq P$ implies that if
  $V(x) \le 1$ then $x^T E x \le \alpha^2$, i.e.,
  $\{ x\in \R^n \, : \, V(x) \le 1 \} \subset \mathcal{E}_\alpha$.
  Hence $\bar{x}=0$ is locally asymptotically stable and the level
  set $\{ x\in \R^n \, : \, V(x) \le 1 \}$ is contained in the ROA $\mathcal{R}$
  (\ak{Theorem 4.1 in \cite{khalil01}}).  Finally, the constraint
  $P \preceq \frac{1}{R^2} I$ implies that if $x^T x \le R^2$, then
  $V(x) \le 1$.  This yields the desired set containment:
  \begin{align*}
    \hat{\mathcal{R}}_R \subset 
    \{ x\in \R^n \, : \, V(x) \le 1 \} \subset  \mathcal{R}.
  \end{align*}
\end{proof}

This theorem provides an inner estimate of the ROA characterized by a
sphere of radius $R$.  A convex optimization can be used to compute
the largest feasible $R$ for given values of
$(E,\alpha,\epsilon)$. Define $\lambda:=\frac{1}{R^2}$ and note that
maximizing $R$ is equivalent to minimizing $\lambda$.  {\tm{Equations
\eqref{eq:roaLMI}}}-\eqref{eq:roaxi} are linear matrix inequalities
(LMIs) in variables $(P,\xi,\lambda)$. The following optimization
is a semidefinite program (SDP):
\begin{align}
  \label{eq:roaSDP}
  \lambda^*:= & \min_{P,\xi,\lambda} \lambda
   \mbox{ subject to } \tm{\eqref{eq:roaLMI}}-\eqref{eq:roaxi}.
\end{align}
An SDP is convex and the global optimum $\lambda^*$ can be computed
efficiently using freely available solvers~\cite{Boyd1994,Boyd2004}. The radius
$R^*=\frac{1}{\sqrt{\lambda^*}}$ provides the largest spherical inner estimate of the ROA for the given local QC region $(E,\alpha)$ and $\epsilon>0$.  The parameter $\epsilon>0$ is chosen to be a
``small'' positive number to ensure $\dot{V}<0$. This term can be dropped if Eq.~\eqref{eq:roaLMI} is
feasible with a strict inequality.  

The main issue with this numerical method is that it requires the
choice of the local QC region in terms of the ellipsoidal shape $E$
and size $\alpha$. If $E=I$, then the QCs are enforced on a sphere of
radius $\alpha$ as shown in Figure~\ref{fig:level_set_diagram}.  A
small value of $\alpha$ will restrict the size of both the Lyapunov
function level set and the spherical ROA inner estimate.  On the other
hand, a large value of $\alpha$ may cause the SDP to be infeasible.
This occurs because the QC bounds on $N(x)$ become more conservative
(less tight) for larger local regions.  A one-dimensional line search
can be used to compute the best value of $\alpha$ for a given local
ellipsoid shape $E$. For example, the SDP in Eq.~\eqref{eq:roaSDP} can
be solved with $E=I$ on a grid of values
$\{\alpha_1,\ldots,\alpha_f\}$.  Each solution yields an \ak{inner}
ROA estimate with radius $R^*(\alpha_i)$. The best $\alpha_i$ is the
one that yields the largest \ak{inner} ROA estimate:
$\max_i R^*(\alpha_i)$.

We can further improve on this \ak{inner} ROA estimate by exploiting
the shape of the ellipsoid as specified by $E$. Unfortunately
Equations \tm{\eqref{eq:roaLMI}}-\eqref{eq:roaxi} are non-convex in
$(P,\xi,R, E,\alpha)$. The first approach, denoted Algorithm A, 
iteratively updates the ellipsoid shape based on the Lyapunov function
obtained from the previous iterate.

\vspace{0.1in}
\noindent{\textbf{Algorithm A:}
\begin{enumerate}
\item \textbf{Initial Estimate:} Define \akr{$M_i(\alpha,E)$} using $E^{(1)}=I$. Find the best $\alpha^{(1)}$ for the given local QCs $E^{(1)}$.  Let ($P^{(1)}, \xi^{(1)}, R^{(1)}$) be the corresponding solutions of the SDP with ($E^{(1)},\alpha^{(1)}$).

\item \textbf{Refinement:} Align the local QC set with the Lyapunov function solution: $E^{(2)} = P^{(1)}$. Find the best $\alpha^{(2)}$ for the updated local QCs $E^{(2)}$.  Let ($P^{(2)}, \xi^{(2)}, R^{(2)}$)  be the corresponding solutions of the SDP with ($E^{(2)},\alpha^{(2)}$)

\item \textbf{Iterate:} Repeat the refinement step with $E^{(i+1)} = P^{(i)}$ to yield ($\alpha^{(i)}, P^{(i)}, \xi^{(i)}, R^{(i)}$). This can be performed a fixed number of iterations or until the radius  $R^{(i)}$ converges.
\end{enumerate}

\begin{figure}[!htb]
    \centering
    \includegraphics[scale = 0.25]{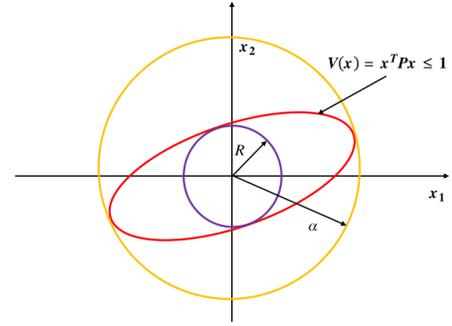}
    \caption{ 2-D visualization of a spherical local region for the QCs corresponding to $E=I$ and $\alpha>0$  (yellow), Lyapunov function level set $\{ x \in \mathbb{R}^n : x^T P x \le 1\}$ (red), and ROA inner estimate $\mathcal{R}_R$ (purple). }
    \label{fig:level_set_diagram}
\end{figure}

}
The optimal solutions
from the first step $(P^{(1)},\xi^{(1)},R^{(1)})$ are also feasible for the second step when
$\alpha^{(2)}=1$. The reason is that the constraint $\frac{1}{\alpha^2} E \le P$ in Eq.~\eqref{eq:roaPineq} holds with equality when using $(P,E,\alpha) = (P^{(1)},P^{(1)},1)$.  Hence the \tm{inner estimate of} ROA cannot shrink at
the second step:  $R^{(2)} \ge R^{(1)}$. Repeating this process
gives a monotonically non-decreasing sequence of spherical inner
estimates for the ROA:  $R^{(i+1)} \ge R^{(i)}$. 
Note that each step of the iterative method has roughly the same
computational cost as the first step. We have to solve one SDP for
each value of $\alpha_i$.

The second approach, denoted Algorithm B below, effectively performs
only a single refinement of the local shape parameter $E$. This
restriction allows the single refinement step to be formulated as a
generalized eigenvalue problem~(GEVP)~(see Eq.~(19) in
\cite{kalur21}).  This will typically reduce the
computational cost, but possibly yield more conservative results
(smaller estimates for $\hat{\mathcal{R}}_R$) as compared to Algorithm~A.  To formulate Algorithm~B, first decompose the quadratic constraint
matrix in Eq.~\eqref{eq:constraint_matrix} into two matrices as
follows:
\begin{align}
    \akr{M_i(\alpha,E)} = \alpha^2 \underbrace{\begin{bmatrix}
    Q_iE^{-1}Q_i & \mathbf{0} \\ \mathbf{0} & \mathbf{0}
    \end{bmatrix}}_{M^{\akr{E}}_i} + \underbrace{\begin{bmatrix} \mathbf{0} & \mathbf{0}\\ \mathbf{0} & -e_ie_i^T \end{bmatrix}}_{{M}^{\akr{e}}_i}.
    \label{eq:Mdecomposition}
\end{align} 

Algorithm B fixes both the shape $E=P^{(1)}$ and Lyapunov function
$P=P^{(1)}$.  This aligns both the local QC ellipsoid shape $E$ with
the level sets of the Lyapunov function.  The local regions for both
are parameterized as
$\{ x\in \R^n \, : \, x^T P^{(1)} x \le \alpha \}$. A sub-problem is
to find the largest local region $\alpha$ over which the local
quadratic constraints are valid and $\dot{V}(x(t))<0$. This is 
formulated by the following optimization:
\begin{align}
\begin{split}
&\min_{\gamma,\xi_0,\ldots,\xi_n} \quad \gamma\\
&\text{subject to}~\xi_{i} \ge 0 \qquad (\text{for}~i=1~\text{to}~n)\\
& \begin{bmatrix} A^TP + PA & P\\
P & \mathbf{0}\end{bmatrix} + \xi_{0}M_0 + \sum_{i=1}^n \xi_{i}{M}^{\akr{e}}_i \prec \tm{\gamma} \sum_{i=1}^n \xi_{i} {M}^{\akr{E}}_i,
\end{split}
\label{eq:GEVP_ROA}
\end{align}
where $\gamma =-\alpha^2$ and $\xi_{i}$~($i=0~\text{to}~n$) are
Lagrange multipliers for the global and local constraints
respectively. It is emphasized that $P=P^{(1)}$ is fixed and not a
decision variable in the optimization.  This is a GEVP~\cite{Boyd1993}
in variables $\alpha^2$, $\xi_0,\ldots,\xi_n$. This one GEVP gives the
largest level set $\alpha^*$ defined by $P=P^{(1)}$ over
which the local quadratic constraints are valid and $\dot{V}(x(t))<0$.
Let $\lambda_{max}(P^{(1)})$ denote the largest eigenvalue of
$P^{(1)}$.  Note that the sphere $\hat{\mathcal{R}}_R$ is contained in
\akr{$\{ x\in \R^n \, : \, x^T P^{(1)} x \le \alpha^{*2} \}$} if and only if
$R \le \frac{\alpha^*}{\sqrt{ \lambda_{max}(P^{(1)})}}$. Thus we can
directly compute the largest radius of the inner ROA estimate
$\hat{\mathcal{R}}_R$ from the optimal
$\alpha^*$.  This leads to our second method to estimate the ROA.

\vspace{0.1in}
\noindent{\textbf{Algorithm B:}}
\begin{enumerate}
\item \textbf{Initial Estimate:} Define \akr{$M_i(\alpha,E)$} using $E^{(1)}=I$. Find the best $\alpha^{(1)}$ for the given local QCs $E^{(1)}$.  Let ($P^{(1)}, \xi^{(1)}, R^{(1)}$) be the corresponding solutions of the SDP with ($E^{(1)},\alpha^{(1)}$).
\item \textbf{Maximize Level Set:} Fix $P=E=P^{(1)}$ and solve the GEVP in Eq.~\eqref{eq:GEVP_ROA} to obtain the maximal level set $\alpha^*$.
\item \textbf{Maximize ROA Inner Estimate:} Select $R^*=\frac{\alpha^*}{\sqrt{ \lambda_{max}(P^{(1)})}}$.
\end{enumerate}
\akr{As a test, we apply algorithm A on the 2-D example in \cite{AmatoConf2007}. We obtain
an inner approximation for the ROA of $R^* = 2.6877$ while
\cite{AmatoConf2007} reports a box of $[-1,1] \times [-2,2]$.  Our disk and the box
have areas of $22.69$ and $8.00$, respectively.}

\section{Numerical Example}
\label{sec:NumericalExample}

We evaluate the proposed analysis methods on two low-order mechanistic
models of transitional flows that were used to demonstrate the QC analysis method in~\cite{liu20}: the 4-state Waleffe-Kim-Hamilton~(WKH)
model~\cite{waleffe95} and the 9-state reduced-order model of \akr{a plane} Couette flow~\cite{Moehlis_2004}.
Both models have the form in Eq.~\eqref{eq:nlsys}, with non-normal linear dynamics and a quadratic lossless nonlinearity.
We note that the linear dynamics' matrix is parameterized by the Reynolds number $Re$: i.e.,~$A=A(Re)$.
Additional details on the specific models used here can be found in~\cite{liu20}.

We begin by using the GEVP in Eq.~\eqref{eq:GEVP_ROA} to estimate the size $\alpha^*$ of the ROA over a range of $Re$.
This is done by applying steps 1 and 2 of Algorithm~B.
%
Figures~\ref{fig:WKH_RoA} and \ref{fig:Couette_RoA} show the results of this analysis (light blue) for the WKH and 9-state Couette flow models, respectively.
These results are compared against ROA estimates based on the quadratic constraints proposed in Liu and Gayme \cite{liu20} (green) and those proposed in Kalur, Seiler, and Hemati \cite{kalur21} (red).
This comparison indicates that the ROA estimate based on refinement of the local QC region in steps 1 and 2 of Algorithm B leads to less conservative estimates on~$\alpha^*$.
Of note here is that although the Liu \& Gayme analysis reduces the conservatism in the analysis relative to the Kalur, Seiler, \& Hemati analysis, the formulation based on ellipsoidal sets reduces conservatism relative to both of these methods by a substantially larger degree for both mechanistic models. 

\begin{figure}[!htb]
    \centering
    \subfigure[WKH Model]{\includegraphics[scale=0.425]{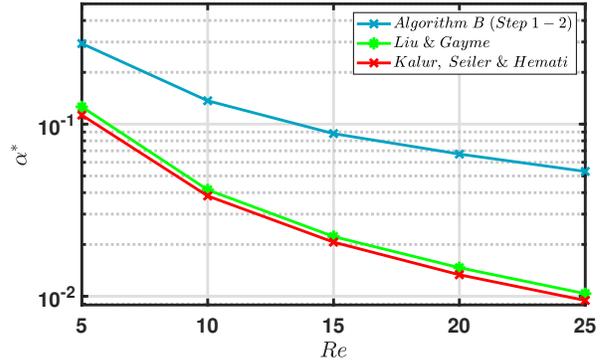}\label{fig:WKH_RoA}}\\
   \subfigure[$9$-state Couette Flow Model]{\includegraphics[scale=0.425]{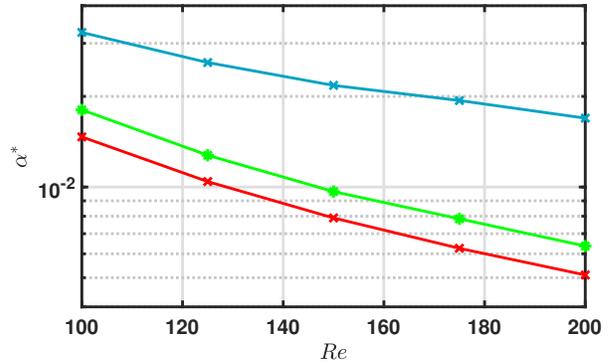}\label{fig:Couette_RoA}}
    \caption{\ak{The ellipsoidal constraints using step 1 and 2 of algorithm B shows significant improvement in region of attraction~(ROA) estimates.}}
    \label{fig:RoA}
\end{figure}




Next, we apply Algorithm A and Algorithm B to estimate the inner approximation~$R^*$ as a function of~$Re$.
This analysis is equivalent to computing a bound on the permissible perturbation amplitude, or sphere of ``safe'' initial conditions.
In Figure~\ref{fig:SOS_R0}, the radius of the largest $\hat{\mathcal{R}}_R$ is denoted as $R^*$  and is obtained from solving Algorithm A and compared with SOS and DAL estimates for the WKH and 9-state Couette flow models, respectively. 
\akr{The DAL method solves a variational problem for the nonlinear optimal perturbation, which is used as a benchmark for comparison.}
\akr{The SOS analysis uses the toolbox available in~\cite{SOSOPT}.}
%
%
To solve Eq.~\eqref{eq:roaSDP} for the WKH model and 9-state models, we use 200 logarithmically spaced values of $\alpha$ between $10^{-5}$ and $10^1$. We compute the ROA estimate using the largest radius obtained on this grid, i.e., $R^*:=\max_i R(\alpha_i^*)$. The results in  Figure~\ref{fig:SOS_R0} show that the ellipsoidal sets improve the estimates of $\hat{\mathcal{R}}_R$ compared to the spherical sets given in~\cite{kalur21,liu20}.
This is true even at the initial iterate, which yields improvements of approximately $4$~times and $2.5$~times for the WKH and 9-state models, respectively.
Additional refinement iterations improve the results even further\akr{. However we set the tolerance for convergence to $10^{-4}$, and also observe only a} marginal improvement after three iterations of Algorithm A~(gray curve).
For the 9-state model, there is an improvement factor of roughly 2.4 and 3.3 using Algorithm A (blue curve) over the QC methods of Liu \& Gayme and Kalur, Seiler \& Hemati.
Additionally, the improvement factor of $R^*$ is  $\approx 3.38$ using Algorithm A as compared to the other two QC constraints for the WKH model.

\begin{figure}[h]
    \centering
        \subfigure[WKH Model]{\includegraphics[scale=0.425]{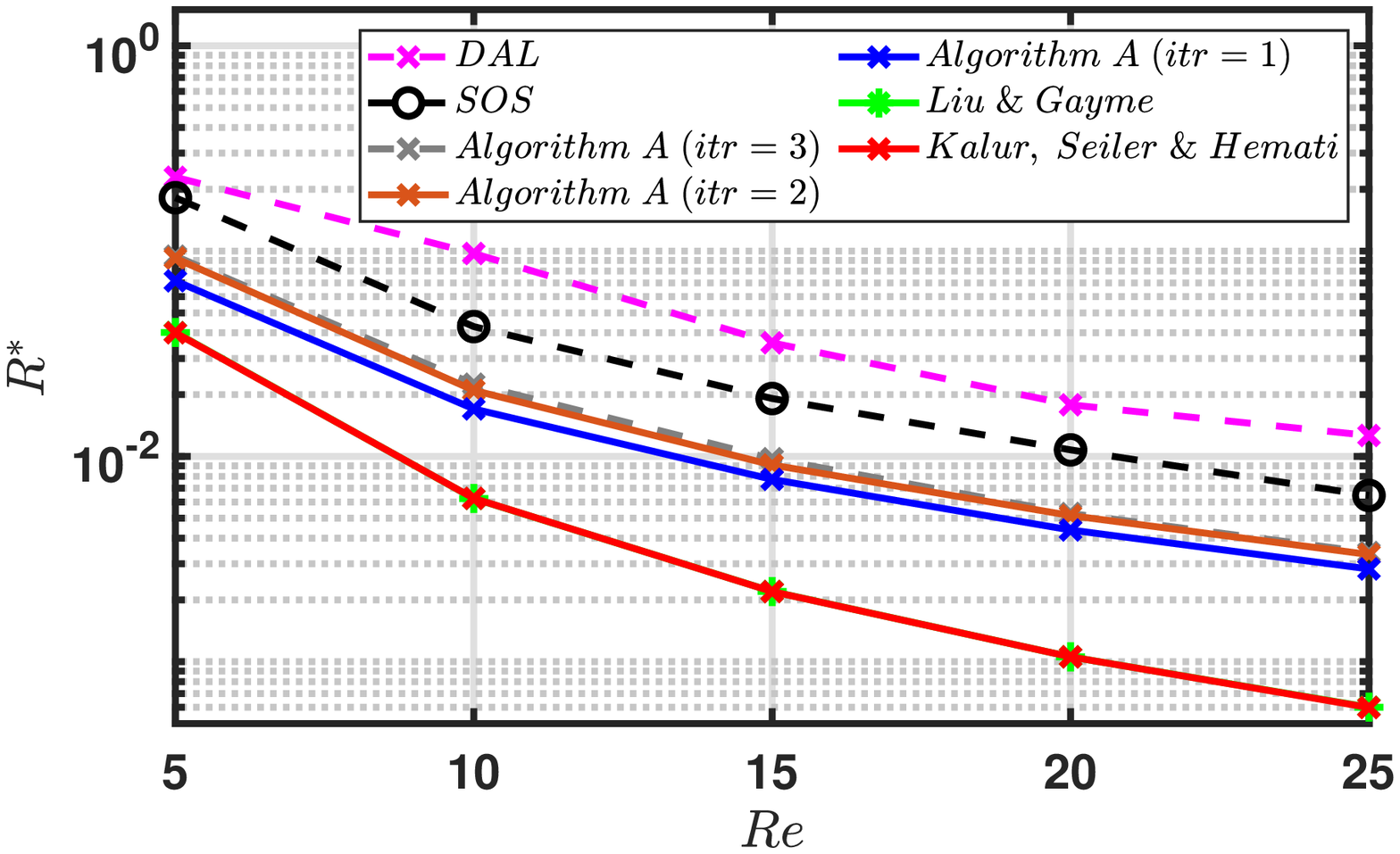}\label{fig:SOS_R0_WKH}}
    \subfigure[$9$-state Couette Flow Model]{\includegraphics[scale=0.425]{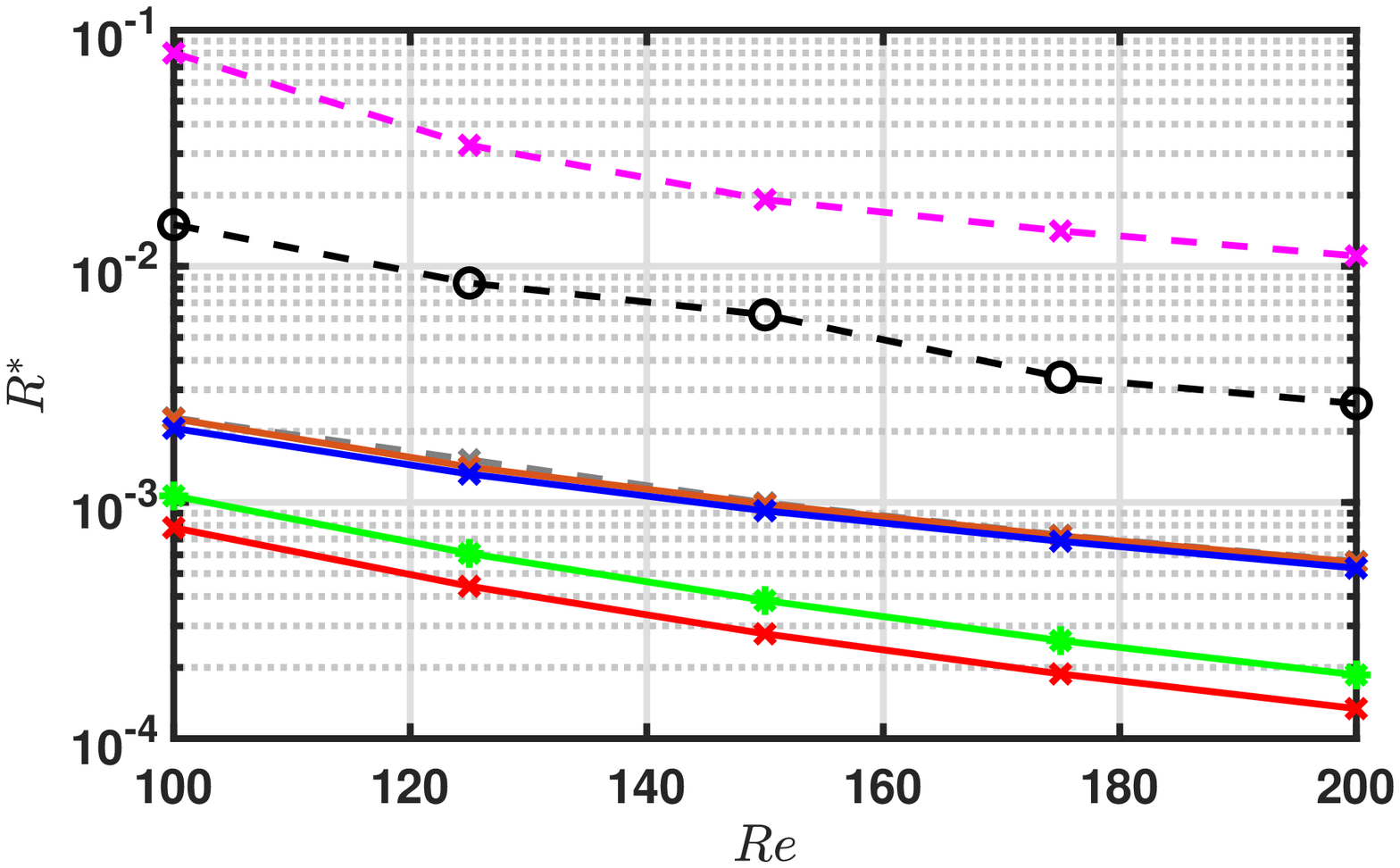}\label{fig:SOS_R0_Couette}}
    \caption{\akr{The inner estimates of ROA obtained using Algorithm~A show improved estimates compared to methods based on spherical sets.}}
    \label{fig:SOS_R0}
\end{figure}

In Figure~\ref{fig:SOS_R0}, we also compare the results obtained using Algorithm A with SOS and DAL methods.
We find that each iteration of Algorithm A reduces the conservatism of the QC estimates, but the inner estimate is still conservative relative to the SOS and DAL methods. 
More specifically, the largest radius $R^*$ obtained from the SOS~(black dashed curve) method and Algorithm A with 1 iteration~(blue curve) differ by an average factor of $\approx 2.45$ and $\approx 6.1$ for the WKH and 9-state models, respectively. 
The differences in the $R^*$ estimates become even greater for the DAL approach, with the DAL estimates (magenta curve) being larger by a factor of $\approx 3.5$ and $\approx 23$ than the Algorithm A estimates for the WKH and 9-state models, respectively.
%
We note that SOS and DAL methods provide superior estimates of $R^*$ because both of these methods use precise information of the nonlinearity and exact equations of motion. 
%
This is in contrast to the QC-based approaches, whereby only input-output properties of the nonlinear terms are used.
%

Next, we assess estimates of $R^*$ using Algorithm~B (see Figure~\ref{fig:Couette_gevp}).
The second step in Algorithm~B avoids the computationally demanding step of solving over a grid of $\alpha$, as is required in Algorithm~A.
Instead, Algorithm~B directly determines the best $\alpha$ for the given shape $E$ and Lyapunov energy matrix $P$, and thus provides an efficient ``one-shot'' approach to estimate $R^*$.
%
Since Algorithm~B does not facilitate further iterations, in general it provides conservative results as compared to Algorithm~A.
However, Algorithm~B substantially reduces conservatism to prior formulations of the QC analysis presented in~\cite{kalur21,liu20}.
In Figure \ref{fig:Couette_gevp}, it can be seen that estimates from Algorithm A and B differ by a factor of roughly 1.16---on average---for the 9-state Couette flow model.
Although not reported here, we made similar observations in our analysis of the WKH model, where the difference was roughly a factor of $1.06$ between Algorithm A and B estimates of $R^*$.
\begin{figure}[h]
    \centering
\includegraphics[scale=0.425]{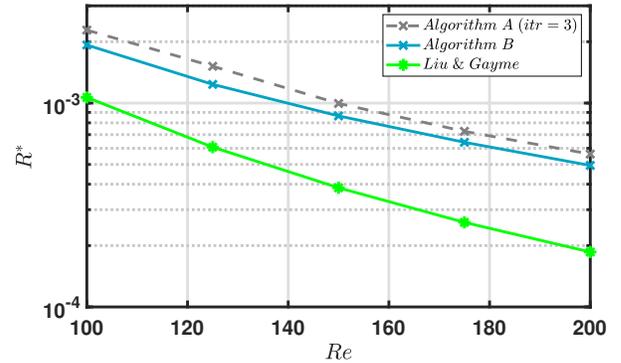}
    \caption{{\akr{The inner estimate of ROA obtained using Algorithm B for the 9-state model is conservative compared to the refinement using Algorithm A. 
    }}}
   \label{fig:Couette_gevp}
\end{figure}


\ak{
\begin{table*}[t]
\caption{Total run-time and average solver time per iteration for calculating $R^*$ by the various methods studied.} 
    \begin{tabular}{|c|c|c||c|c|}
    \hline
    \multirow{2}{*}{Method}  & \multicolumn{2}{|c||}{Run-time \akr{to convergence}~(secs)} & \multicolumn{2}{c|}{Avg. solver time \akr{for one} iteration~(secs)}\\ \cline{2-5}
    & WKH & 9-state & WKH & 9-state  \\\cline{2-5}
   \hline
  Algorithm A \akr{(solver: Mincx Matlab)} &  2.89 &  63.48 &  0.76$\times 10^{-2}$ &  0.15   \\
   \hline
  Algorithm B \akr{(solver Mincx \& gevp~Matlab)} & 1.85 & 37.99 &  2.3$\times 10^{-2}$  & 0.49 \\
  \hline
   DAL & 9.82 & 137.6 &  8.64 $\times 10^{-4}$ & 0.18 $\times 10^{-2}$ \\
   \hline
   SOS & 116.8 & 5.82$\times10^4$ & 2.9  & 1.45$\times10^3$\\ \hline
    \end{tabular}
    \label{table:1}
\end{table*}
}
Finally, we assess the computational run-time performance of the various methods investigated in this study.
All computations were performed on an ASUS ROG M15 laptop with Intel 2.6 GHz i7-10750H CPU and a 16 GB RAM.
Overall both Algorithm A and B proposed in this paper require less total run-time compared to DAL and SOS methods.
This savings becomes especially apparent in analyzing the 9-state model.
Although the SOS and DAL methods yield more accurate estimates, these methods scale poorly with the state dimension compared to the QC analysis methods.
The SOS method for WKH has a wall-time of about $116.8$ seconds as compared to $5.82\times10^4$ seconds for the 9-state model. Similarly, the solver run times for each iteration of WKH model is $2.9$ seconds compared to $1.45\times10^3$ seconds for the 9-state model. Thus, in case of the SOS method, roughly doubling the states results in the  total computation time increasing by a factor of $\approx 500$. 
In contrast, the total run-time of the QC-based Algorithm A increased by a factor of roughly $20$ between the 4-state WKH model and the 9-state model.  For the 9-state model, when we compare the run-time for Algorithm A to the SOS method, we see that the QC method is approx $900$ times faster. 
%
%
We note that the run-time for the DAL method appears to increase by a factor of roughly 15 when going from the 4-state WKH model to the 9-state model, which actually seems to scale better than even the QC method;
%
%
however, it is important to note that the DAL method can be sensitive to the final simulation time, perturbation size, tolerances, etc. Thus, tuning the DAL method can be a time intensive process, especially when system parameters (e.g.,~$Re$) are changed.
The time required to tune the DAL process to obtain the precise estimates reported in this study is not reflected in the times listed in Table~\ref{table:1}.
%
Overall, we conclude that the QC-based Algorithms A and B require less end-to-end time than SOS and DAL methods, and yield $R^*$ solutions that are approximately within one order of magnitude of the SOS and DAL estimates.

\section{Conclusions}
In this work, we have proposed an improvement to  the quadratic constraint (QC) framework
for nonlinear fluid flow analysis.
This was done by generalizing the local QCs from spherical sets~\akr{(proposed in \cite{kalur21,kalurAIAA2020,liu20})} to ellipsoidal sets, which reduced conservatism and improved estimates of the ROA. 
Additionally, we proposed and investigated two algorithms for performing the ROA analysis. 
The less conservative but more computationally demanding algorithm---Algorithm A---iteratively refines the solution by solving a sequence of semi-definite programs.
In contrast, the more computationally efficient algorithm---Algorithm B---solves a single generalized eigenvalue problem~(GEVP) and yields estimates of the ROA and permissible perturbation amplitude in a single pass.
Both Algorithms A and B were found to outperform the  QC analysis methods proposed in~\cite{kalur21} and \cite{liu20} in terms of accuracy.
Algorithm~B did so at no additional computational cost over these prior QC-based analysis methods.
Both of the proposed algorithms surpassed prevailing SOS and DAL methods in terms of computational run-time.
Although the proposed methods did not attain the same degree of accuracy as the computationally demanding SOS and DAL methods, both Algorithms A and B estimated results on the same order of magnitude as DAL and SOS for the models considered here.
It may still be possible to refine the QC method beyond what we have presented in this study.
Future work may benefit from incorporating additional constraints to refine the proposed QC analysis even further.
\section*{Acknowledgments}

This material is based upon work supported by the Army Research Office under grant number W911NF-20-1-0156. 
%
MSH acknowledges support from the Air Force Office of Scientific Research under award number FA 9550-19-1-0034 
and the National Science Foundation under grant
number CBET-1943988.

\bibliographystyle{IEEEtran}
\bibliography{ROAEstimationForFlows}

\end{document}